\documentclass[12pt]{article}
\usepackage{amssymb,amsmath,amsthm,dsfont}

\theoremstyle{plain}
\newtheorem{thm}{THEOREM}
\newtheorem{lem}{LEMMA}
\theoremstyle{definition}

\newcommand{\A}{{\mathsf A}}
\newcommand{\B}{{\mathsf B}}
\newcommand{\sA}{{\mathsf a}}
\newcommand{\sB}{{\mathsf b}}
\newcommand{\C}{{\mathbb C}}
\newcommand{\M}{{\mathsf M}}
\newcommand{\W}{{\mathsf W}}
\newcommand{\N}{{\mathbb N}}
\newcommand{\R}{{\mathbb R}}
\newcommand{\s}{{\mathcal S}}
\newcommand{\Tr}{{\rm Tr\, }}
\newcommand{\id}{{\mathds 1}}
\newcommand{\la}{{\lambda}}
\newcommand{\F}{{\mathcal{F}}}

\begin{document}

\title{Further implications of the Bessis-Moussa-Villani conjecture}
\author{\vspace{5pt} Elliott H.~Lieb$^{*}$ and Robert Seiringer$^{
\dagger}$\\
\vspace{-4pt}\small{$^{*}$Departments of Mathematics and Physics,
Jadwin Hall,} \\
\small{Princeton University, P.~O.~Box 708, Princeton, New Jersey
  08544, USA}\\ \vspace{-4pt}\small{$^{\dagger}$Department of Mathematics and Statistics, McGill University,}\\ \small{805 Sherbrooke St. West, Montreal QC H3A0B9, Canada}}
\date{\small June 3, 2012}

\maketitle

\renewcommand{\thefootnote}{$ $}
\footnotetext{\copyright\, 2012 by the authors. This paper may be reproduced, in its
entirety, for non-commercial purposes.}

\begin{abstract}
  We find further implications of the BMV conjecture, which states
  that for hermitian matrices $\A$ and $\B$, the function
  $\lambda\mapsto\Tr \exp(\A -\la \B)$ is the Laplace transform of a
  positive measure.
\end{abstract}

\bigskip
Bessis, Moussa and Villani conjectured \cite{bessis} that for hermitian matrices $\A$
and $\B$, $\lambda\mapsto\Tr \exp(\A -\la \B)$ is the Laplace transform
of a positive measure. Previously we showed that this property is
equivalent to several other trace inequalities, most notably that the
polynomials $\F_p(\lambda) = \Tr (\A + \lambda \B)^p$ for $p\in\N$ all have positive coefficients
when $\A$ and $\B$ are positive. Since a proof of the BMV conjecture has
recently been put forward by H. Stahl \cite{stahl}, it seems
worthwhile to find other implications.

Here we prove two things. Our first result is that the function $\lambda \in \R_+\mapsto
\Tr (\A+\lambda\B)^p$ for general $p>0$ has positive derivatives up to order $\lceil p\rceil$,
the largest integer not less than $p$. Moreover, taking further
derivatives, one obtains alternating signs for the derivatives. Our
second result is that our earlier theorem on $\F_p(\lambda)$ has a generalization from
sums of eigenvalues to elementary symmetric functions of eigenvalues.

\section{Our previous results}

The following theorem was proved in \cite{LS}.

\begin{thm}\label{T1}
For fixed $n$ 
let $\A$ and $\B$ denote arbitrary hermitian $n\times n$ matrices over $\C$,
and let $\lambda\in\R$. 
The following statements are equivalent:
\begin{itemize}
\item[(i)] For all $\A$ and $\B$ positive, and all $p\in\N$, the polynomial
$\lambda\mapsto\Tr(\A+\lambda\B)^p$ has only non-negative coefficients.
\item[(ii)] For all $\A$ hermitian and $\B$ positive, $\lambda\mapsto\Tr
\exp{(\A-\lambda\B)}$ is the Laplace transform of a positive
measure supported in $[0,\infty)$.
\item[(iii)] For all $\A$ positive definite and $\B$ positive,
and all $p \geq 0$,  $\lambda\mapsto \Tr (\A+\lambda \B)^{-p}$ is the
Laplace transform of a positive measure supported in $[0,\infty)$.
\end{itemize}
\end{thm}

We remark that items (i) and (iii) can be combined to the statement
that $\F_p(\lambda) = \Tr (\A + \lambda\B)^p$ has positive derivatives
when $p$ is a positive integer, and derivatives of alternating sign
when $p$ is negative. 

The positivity property in items (ii) and (iii)
follow from Bernstein's Theorem \cite{don} if the functions have alternating
derivatives for all $\lambda\geq 0$. Since the statement above
involves arbitrary $\A$, it suffices to check the alternating
derivative property at $\lambda = 0$.

\section{Extension to general $p\in \R$}

\begin{thm}\label{T2}
  Item (ii) in Theorem~\ref{T1} has the following consequences. For
  all $\A$ and $\B$ positive, $p\in \R$, we have
\begin{itemize}
\item [a)] For $1\leq r\leq \lceil p\rceil$, \ $\frac{d^r}{d\lambda^r} \F_p(\lambda) \geq 0$ for $\lambda \geq 0$.
\item [b)] For $r \geq \lceil p\rceil$ and $p>0$, \  $(-1)^{r-\lceil p \rceil}\frac{d^r}{d\lambda^r} \F_p(\lambda) \geq 0$ for $\lambda \geq 0$.
\item [c)] For $r\geq 1$ and $p\leq 0$,  $(-1)^{r}\frac{d^r}{d\lambda^r} \F_p(\lambda) \geq 0$ for $\lambda \geq 0$,
\end{itemize}
where $\F_p(\lambda)= \Tr(\A+\lambda \B)^p$. 
\end{thm}

We remark that item c) follows directly follows directly from item (iii) in Theorem~\ref{T1}, we included it in Theorem~\ref{T2} for completeness. Our proof of item b) does not require item (ii), in fact, and we shall give that first.

\begin{proof}[Proof of Theorem~\ref{T2}(b)]
Let $s= \lceil p \rceil - p$. We can assume that $s>0$. We start with the integral representation
\begin{equation}\label{intrep}
 (\A+\lambda \B)^p =  \frac{\sin(\pi s)}{\pi} \int_0^\infty \frac{ (\A+ \lambda \B)^{\lceil p\rceil}}{\A+ \lambda \B + t} \,t^{-s}\, dt\,.
\end{equation}
Using the binomial theorem, we have
$$
(\A+ \lambda \B)^{\lceil p\rceil} = \sum_{j=0}^{\lceil p \rceil} \binom{\lceil p \rceil}{j} (-t)^{j} (\A + \lambda \B + t)^{\lceil p \rceil - j}
$$
for $t>0$. In particular, since $r \geq \lceil p \rceil$, only the term with $j=\lceil p\rceil$ contributes to the $r^{\rm th}$ derivative of the integrand in (\ref{intrep}), i.e.,
$$
\frac{d^r}{d\lambda^r} \,\frac{ (\A+ \lambda \B)^{\lceil p\rceil}}{\A+ \lambda \B + t} = \frac{d^r}{d\lambda^r} (-t)^{\lceil p \rceil} \frac 1{\A + \lambda \B + t}\,.
$$
Hence
$$
\frac{d^r}{d\lambda^r}\, \Tr (\A + \lambda \B) ^ p = (-1)^{\lceil p \rceil}  \frac{\sin(\pi s)}{\pi}  \int_0^\infty \, \Tr\,  \frac{d^r}{d\lambda^r}\, \frac 1{\A + \lambda \B + t}\, t^{p} \, dt\,.
$$
Using the resolvent identity, we have
$$
 \frac{d^r}{d\lambda^r}\, \frac 1{\A + \lambda \B + t} = \frac {(-1)^r}{\A + \lambda \B + t} \left( \B \frac 1{\A + \lambda \B + t} \right)^r\,,
$$
from which we easily conclude that
$$
(-1)^r\,  \Tr\, \frac{d^r}{d\lambda^r}\, \frac 1{\A + \lambda \B + t} \geq 0\,,
$$
which completes the proof.
\end{proof}

For the proof of part a) of Theorem~\ref{T2}, we shall need the following lemma.

\begin{lem}\label{L1}
Let $\sA$ and $\sB$ be hermitian $n\times n$ matrices over $\C$,
with $\sA$ positive definite. Define $\A=\sA^{-1}$
and $\B=\sA^{-1/2}\sB \sA^{-1/2}$, and let $\lambda\in\R$. For all
$p \in \C$ and $r\in \N$
\begin{equation}\label{equiv}
(p+r) \left.\frac {d^r}{d\lambda^r} \Tr \frac 1{\left(\sA+\lambda
\sB\right)^{p}}\right|_{\lambda=0}= p (-1)^r \left.
\frac {d^r}{d\lambda^r} \Tr \left(\A+\lambda
\B\right)^{p+r}\right|_{\lambda=0} \ .
\end{equation}
\end{lem}

\begin{proof}
The proof of (\ref{equiv}) for $p\in\N$ was given in \cite{LS}; we include it verbatim in the appendix for completeness. 

Both sides of (\ref{equiv}) are entire functions of $p$. Let $f(p)$
denote the left side minus the right side. We have just noted that
$f(p)=0$ for $p\in \N$. Moreover, we can find an $a>0$ such that the
function $p\mapsto f(p) e^{-ap}$ is bounded for $\Re p \geq 0$. It
then follows from Carlson's theorem (see \cite{hardy}) that $f$ is
identically zero in the half-space $\Re p \geq 0$, and hence for all
$p\in \C$.
\end{proof}

\begin{proof}[Proof of Theorem~\ref{T2}(b)]
  As remarked after Theorem~\ref{T1}, it is sufficient to prove the
  statement for $\lambda = 0$.  We can assume that $p>r$, the statement
  is trivial for $p=r$. From the identity (\ref{equiv}) with $p$
  replaced by $p-r$, we have
\begin{equation}\label{aa}
\left. \frac{d^r}{d\lambda^r} \Tr (\A+ \lambda \B)^p \right|_{\lambda = 0} = (-1)^r \frac{p}{p-r} \left. \frac{d^r}{d\lambda^r} \Tr (\sA + \lambda \sB)^{r-p} \right|_{\lambda=0}\,.
\end{equation}
By item (iii) of Theorem~\ref{T1}, the function $\lambda \mapsto \Tr (\sA + \lambda \sB)^{r-p}$ has alternating derivatives for $p>r$, and hence the right side of (\ref{aa}) is positive. 
\end{proof}

\section{Extension to elementary symmetric functions of eigenvalues}

We now return to the polynomials $\F_p(\lambda) = \Tr (\A + \lambda
\B)^p$ for {\em positive integer} $p$. The coefficient of $\lambda^k$
in this polynomial is the trace of a sum of $p$-letter words in two
letters, with $\A$ appearing $p-k$ times, and $\B$ appearing $k$
times. It is known that the trace of an individual words need
not be positive \cite{johnson}, but the sum of the traces is, according to
Theorem~\ref{T1} and the BMV conjecture. 

Instead of traces, which involve sums of eigenvalues, let us consider
determinants, which are the products of all the eigenvalues. It is
clear that the determinant of the sum of all words for given $p$ and
$k$ need not be positive, as the example $p=2$ and $k=1$ shows;
namely, while $\Tr (\A\B + \B\A)$ is positive, the determinant $\det
(\A \B + \B\A)$ need not be \cite{ball}. On the other hand, each
individual determinant is clearly positive, since it is equal to
$(\det \A)^{p-k} (\det \B)^k$. This suggests that general elementary
symmetric functions of the eigenvalues of the sum of all words need
not be positive, while the sum of the elementary symmetric functions
of the eigenvalues of the individual words should be positive. We prove that here.

The elementary symmetric functions of the eigenvalues,
of which the trace and the determinant are two special cases, are given by 
$$
e_j(\lambda_1,\dots,\lambda_n) = \sum_{1\leq i_1 < i_2 < \dots < i_j \leq n} \prod_{k=1}^j \lambda_{i_k}
$$
for $1\leq j \leq k$. This number $e_j$ is also equal to the sum of
the principal subdeterminants of order $j$ of a matrix $\M$ with
eigenvalues $\lambda_1,\dots,\lambda_n$, and we shall also write
$e_j(\M)$ for short. Another way to think of $e_j(\M)$ is as the trace of the
$j^{\rm th}$ adjugate of $\M$, which is the $j$-fold
anti-symmetric tensor product of $\M$ with itself, $\M\wedge \cdots\wedge \M$.

\begin{thm}\label{T3}
Item (i) in Theorem~1, assumed to hold for all $n\in \N$, has the following consequence. For all $\A$ and $\B$ positive, $1\leq k \leq  p$ and $1\leq j \leq n$, 
$$
\sum_{i=1}^{\binom{p}{k}} e_j(\W_i) \geq 0
$$
where the  $\W_i$ denote all words of length $p$ with $k$ letters $\B$ and $p-k$ letters $\A$.   
\end{thm}

As remarked above, in general it is false that $e_j( \sum_i \W_i) \geq 0$, except for $j=1$, where $\sum_i e_j(\W_i) = e_j(\sum_i \W_i)$.

\begin{proof}
  We apply item (i) in Theorem~\ref{T1} to the $j$-fold antisymmetric
  tensor products of $\A$ and $\B$, respectively. Since the tensor
  product of a product of matrices equals the product of the
  individual tensor products, the theorem follows immediately.
\end{proof}

In a similar fashion, we can prove the following.

\begin{thm}\label{T4}
Item (ii) in Theorem~1, assumed to hold for all $n\in \N$, has the following consequences. 
\begin{itemize}
\item [a)]
For all hermitian $\A$ and positive $\B$, $k\geq 1$ and $1\leq j \leq n$, 
$$
\int\limits_{s_i\geq 0,\, \sum_{i=1}^{k+1} s_i = 1} e_j\left( e^{s_1 A} B e^{s_2 A} B \dots B e^{s_{k+1} A} \right) \, ds_1\cdots ds_{k+1} \geq 0
$$
\item [b)] For all hermitian $\A$ and positive $\B$, and $1\leq j \leq n$, 
$$
\lambda \mapsto e_j\left( e^{\A - \lambda \B} \right)
$$
is the Laplace transform of a positive measure supported in $[0,\infty)$.
\end{itemize}
\end{thm}

\begin{proof} Item a) follows by applying item (ii) with $\A$ replaced by
$$
\alpha = \left(\A \wedge \id \wedge \cdots \wedge \id \right) + \left(\id \wedge \A \wedge \id \cdots \wedge \id\right)  + \dots + \left( \id \wedge \cdots \wedge \id \wedge \A\right)
$$
and $\B$ replaced by $\beta = \B\wedge \B\wedge \cdots \wedge \B$. Note that $e^\alpha = e^\A \wedge e^\A \wedge \cdots\wedge e^\A$. The $k^{\rm th}$ derivative of $e^{\alpha-\lambda \beta}$ with respect to $\lambda$ is equal to 
$$
(-1)^k \int\limits_{s_i\geq 0,\, \sum_{i=1}^{k+1} s_i = 1} e^{s_1 \alpha} \beta e^{s_2 \alpha} \beta \dots \beta e^{s_{k+1} \alpha}  \, ds_1\cdots ds_{k+1} 
$$
and hence the statements follows.

To obtain item b), we replace $\A$ by $\alpha$ and $\B$ by 
$$
\gamma = \left(\B \wedge \id \wedge \cdots \wedge \id\right) + \left(\id \wedge \B \wedge \id \cdots \wedge \id\right)  + \dots + \left( \id \wedge \cdots \wedge \id \wedge \B\right)\,.
$$
Then $e^{\alpha - \lambda \gamma} = e^{\A-\lambda \B} \wedge e^{\A - \lambda \B} \wedge \cdots \wedge e^{\A-\lambda \B}$. 
\end{proof}

The case $j=n$ in item b) of Theorem~\ref{T4} follows immediately from the fact that 
$$
\det e^{\A - \lambda \B} = e^{\Tr \A - \lambda \Tr \B}\,,
$$
and holds even without the assumption of the Theorem.

Note that 
$$
e_2 ( \M ) = \frac 12 \left( \left( \Tr \M \right)^2 - \Tr \M^2 \right)\,.
$$
For $\M  =e^{\A - \lambda \B}$, both $(\Tr \M)^2$ and $\Tr \M^2$ are the Laplace transform of a positive measure. It is remarkable that also their difference has this property! Similar conclusions can be drawn for general $j\geq 2$.

\appendix

\section{Appendix: Proof of Lemma~\ref{L1} for $p\in\N$}

By induction it is easy to show that
\begin{equation}
\frac {d^r}{d\lambda^r} \left(\A+\lambda \B\right)^{p+r}=r!
\sum_{\substack{0\leq i_1,\dots,i_{r+1}\leq p \\ \sum_j i_j=p }}
(\A+\lambda \B)^{i_1} \B \cdots \B (\A+\lambda\B)^{i_{r+1}} \ .
\end{equation}
By taking the trace at $\lambda=0$ we obtain
\begin{equation}
I_1\equiv \left. \frac {d^r}{d\lambda^r} \Tr \left(\A+\lambda
\B\right)^{p+r}\right|_{\lambda=0}=r! \sum_{\substack{0\leq
i_1,\dots,i_{r+1}\leq p
\\ \sum_j i_j=p }} \Tr \A^{i_1} \B \cdots \B
\A^{i_{r+1}}   \ .
\end{equation}
Moreover, by similar arguments,
\begin{equation}
\frac {d^r}{d\lambda^r} \frac 1{\left(\sA+\lambda
\sB\right)^{p}}=(-1)^r r! \sum_{\substack{1\leq
i_1,\dots,i_{r+1}\leq p
\\ \sum_j i_j=p+r }} \frac 1{(\sA+\lambda \sB)^{i_1}} \sB \cdots \sB
\frac 1{(\sA+\lambda\sB)^{i_{r+1}}} \ .
\end{equation}
By taking the trace at $\lambda=0$ and using cyclicity, we get
\begin{equation}
I_2\equiv\left.\frac {d^r}{d\lambda^r} \Tr \frac
1{\left(\sA+\lambda \sB\right)^{p}}\right|_{\lambda=0}=(-1)^r r!
\sum_{\substack{0\leq i_1,\dots,i_{r+1}\leq p-1
\\ \sum_j i_j=p-1 }} \Tr \A\, \A^{i_1} \B \cdots \B
\A^{i_{r+1}}\, \ .
\end{equation}
We have to show that
\begin{equation}
I_2=\frac{p}{p+r} (-1)^r I_1 \ .
\end{equation}
To see this we rewrite $I_1$ in the following way. Define $p+r$
matrices $\M_j$ by
\begin{equation}
\M_j=\left\{\begin{array}{ll}\B & {\rm for\ }1\leq j\leq r \\ \A &
{\rm for\ } r+1\leq j\leq r+p\ .\end{array}\right.
\end{equation}
Let $\s_n$ denote the permutation group. Then
\begin{equation}
I_1=\frac 1{p!} \sum_{\pi\in\s_{p+r}} \Tr \prod_{j=1}^{p+r}
\M_{\pi(j)} \ .
\end{equation}
Because of the cyclicity of the trace we can always arrange the
product such that $\M_{p+r}$ has the first position in the trace.
Since there are $p+r$ possible locations for $\M_{p+r}$ to appear in the
product above, and all products are equally weighted,  we get
\begin{equation}
I_1=\frac {p+r}{p!} \sum_{\pi\in\s_{p+r-1}} \Tr \A
\prod_{j=1}^{p+r-1} \M_{\pi(j)} \ .
\end{equation}
On the other hand,
\begin{equation}
I_2=(-1)^r \frac 1{(p-1)!} \sum_{\pi\in\s_{p+r-1}} \Tr \A
\prod_{j=1}^{p+r-1} \M_{\pi(j)}\ ,
\end{equation}
so we arrive at the desired equality.

\bigskip
\noindent {\it Acknowledgments.} We are grateful to P. Landweber for
valuable discussion on this topic, especially about elementary
symmetric functions. Partial financial support by U.S. NSF grant PHY-0965859
(E.H.L.) and the NSERC (R.S.) is gratefully acknowledged.

\end{document}